\documentclass[11pt]{article}

\usepackage{verbatim}
\usepackage{textcomp}
\usepackage{latexsym}
\usepackage{amsmath}
\usepackage{amsfonts}
\usepackage{amssymb}
\usepackage{amsthm}
\usepackage{amsmath}

\newcommand{\F}{\mathbb{F}}

\newtheorem{definition}{Definition} 
\newtheorem{corollary}{Corollary}

\newtheorem{lemma}{Lemma} 
\newtheorem{theorem}{Theorem}

\usepackage[normalem]{ulem}

\usepackage{graphicx}
\begin{document}
\title{On the Hardness of the Determinant: Sum of Regular Set-Multilinear Circuits}
%
%
\author{S. Raja\thanks{IIT Tirupati, India,
    \texttt{email: raja@iittp.ac.in}} \and Sumukha Bharadwaj G V\thanks{IIT Tirupati, India,
    \texttt{email: cs21d003@iittp.ac.in}}}
%
%
%

%
\date{}
\maketitle              
\begin{abstract}
	In this paper, we study the computational complexity of the commutative determinant polynomial computed by a class of set-multilinear circuits which we call \emph{regular set-multilinear circuits}. Regular set-multilinear circuits are commutative  circuits with a restriction on the order in which they can compute polynomials. A regular circuit can be seen as the commutative analogue of the \emph{ordered circuit} defined by Hrubes,Wigderson and Yehudayoff \cite{HWY10}. We show that if the commutative determinant polynomial has small representation in the sum of constantly many regular set-multilinear circuits, then the commutative permanent polynomial also has a small arithmetic circuit. 
\end{abstract}
\section{Introduction}
\emph{Arithmetic circuit complexity} studies the complexity of computing polynomials using arithmetic
operations. Arithmetic circuits are a natural computational model for
computing and describing polynomials. Arithmetic circuit is a directed acyclic graph with internal nodes labeled by + or $\times$, and leaves labeled by either variables or elements from a underlying field $\mathbb{F}$. The complexity measures associated with arithmetic circuits are size, which measures number of gates in the circuit, and depth, which measures length of the longest path from a leaf to the output gate in the circuit. Two important examples of polynomial family are the determinant and the permanent polynomials. The determinant polynomial is ubiquitous in linear algebra, and it can
be computed by polynomial-sized arithmetic circuits (see e.g., \cite{Ber84}). On the other
hand, the permanent of 0/1 matrices is \#P-complete \cite{Val79b}, where \#P corresponds to the counting class in the world of Boolean complexity classes. Thus, it is believed that, over fields of characteristic different from 2, the permanent $PERM = (PERM_n)$ polynomial family cannot be computed by any polynomial-sized circuit family. A central open problem of the field is proving super-polynomial size lower bounds for arithmetic circuits that compute the permanent polynomial $PERM_n$. Motivated by this problem, Valiant, in his seminal work \cite{Val79a}, defined the arithmetic analogues of P and NP: denoted by VP and VNP. Informally, VP consists of multivariate (commutative) polynomials
that have polynomial size circuits. Valiant showed that $PERM$ is VNP-complete
w.r.t. projection reductions. Thus, $VP\neq VNP$ iff $PERM_n$ requires arithmetic circuits of size super-polynomial in $n$.

Set-multilinear circuits are introduced in the work of \cite{NW95}.
Let $\mathbb{F}$ be a field and $X=X_1\sqcup X_2\sqcup\dots\sqcup X_d$ be a
partition of the variable set $X$. A \emph{set-multilinear polynomial}
$f\in\mathbb{F}[X]$ w.r.t.\ this partition is a homogeneous degree $d$
multilinear polynomial such that every nonzero monomial of $f$ has
exactly one variable from $X_i$, for all $1\le i\le d$.
Some of the well-known polynomial families like the permanent $PERM_n$ and the determinant $DET_n$, are set-multilinear. The variable set is
$X=\{x_{ij}\}_{1\le i,j\le n}$ and the partition can be taken as the
row-wise partition of the variable set. I.e.\ $X_i=\{x_{ij}\mid 1\le
j\le n\}$ for $1\le i\le n$. In this work, we study the set-multilinear circuit complexity of the determinant polynomial
$DET_n$. A \emph{set-multilinear arithmetic circuit} $C$ computing $f$
w.r.t.\ the above partition of $X$, is a directed acyclic graph such
that each in-degree $0$ node of the graph is labeled with an element
from $X\cup \F$. Each internal node $v$ of $C$ is of in-degree $2$, and
is either a $+$ gate or $\times$ gate. With each gate $v$ we can
associate a subset of indices $I_v\subseteq [d]$ and the polynomial
$f_v$ computed by the circuit at $v$ is set-multilinear over the
variable partition $\bigsqcup_{i\in I_v} X_i$. If $v$ is a $+$ gate then
for each input $u$ of $v$,  $I_u=I_v$. If $v$ is a $\times$ gate with
inputs $v_1$ and $v_2$ then $I_v= I_{v_1}\sqcup I_{v_2}$. Clearly, in
a set-multilinear circuit every gate computes a set-multilinear
polynomial (in a syntactic sense). The output gate of $C$ computes the polynomial $f$, which is set-multilinear over the variable partition $\bigsqcup_{i\in [d]} X_i$. The \emph{size} of $C$ is
the number of gates in it and its \emph{depth} is the length of the
longest path from an input gate to the output gate of
$C$. Additionally, a set-multilinear circuit $C$ is called a \emph{set-multilinear formula} if out-degree of every gate is bounded by $1$.

Set-multilinear arithmetic circuits are a natural model for computing set-multilinear polynomials. It can be seen that each set-multilinear polynomial can be computed by a set-multilinear arithmetic circuit.
For set-multilinear formulas, super-polynomial size lower bounds are known \cite{Raz09}. Super-polynomial lower bounds  for a class of set-multilinear ABPs computing the determinant $DET_n$ is shown in \cite{AR16}. 
It is known that proving super-polynomial lower bound result for general set-multilinear circuits computing the permanent polynomial $PERM_n$  would imply that $PERM_n$ requires super-polynomial size non-commutative arithmetic circuits, and this is an open problem for over three decades. 
Non-commutative circuits are a restriction on the computational power of circuits. Though non-commutative circuits compute non-commutative polynomials, one can study what is the power of commutativity in computing the $DET_n$ polynomial.  
Noncommutative arithmetic circuit models are well studied, see e.g., \cite{N91,AS10,HWY10}. In \cite{AS10}, it was shown that computing the non-commutative determinant polynomial is as hard as computing the commutative permanent polynomial.

\subsection{Our Results} 
To explain our results, we first define the computational model that we study. Let $S_n$ denote the set of all permutations over the set $\{1,2...,n\}$.

\begin{definition}[Regular Set-Multilinear Circuits]
	Let $X=X_1\sqcup X_2\sqcup\dots\sqcup X_d$ be a partition of the variable set $X$. Let $\sigma \in S_d$. 
	A set-multilinear circuit $C$ that computes a set-multilinear polynomial $f\in F[X]$ w.r.t the above partition  is called \emph{regular set-multilinear circuit w.r.t $\sigma \in S_d$}, if every gate $v$ in $C$ is associated with an interval $I_v$ w.r.t  $\sigma \in S_d$. In other words, $\sigma \in S_d$ defines an ordering $(\sigma(1),\sigma(2),\cdots,\sigma(d))$ and every gate $v$ in $C$ is associated with an interval $I_v$ w.r.t $\sigma$-ordering $(\sigma(1),\sigma(2),\cdots,\sigma(d)).$ 
\end{definition}

Let $C$ be a regular set-multilinear circuit w.r.t $\sigma$ computing a commutative polynomial $f$ of degree $d$. 
Let $v$ be a gate in $C$ computing the polynomial $f_v$ of degree $k$. By definition, $f_v$ is a  set-multilinear polynomial w.r.t $I_v=[\sigma(i),\sigma(i+1),\cdots,\sigma(i+k)]$, where $i<=d-k$. Let $order(f_v)=I_v=(\sigma(i),\sigma(i+1),\cdots,\sigma(i+k))$.

Since for each gate $v$ in $C$, $I_v$ can be viewed as an interval w.r.t $\sigma \in S_d$, the two children $u$ and $w$ of $v$ can be designated as left and right child. In particular, for each product gate $v$ with children $u$ and $w$ such that $I_v=I_u \sqcup I_w$, we refer to $u$ as the left child of $v$, and $w$ as the right child of $v$. 

We make the following observations about regular set-multilinear circuits:
\begin{itemize}
	\item If $v$ is an input gate (leaf node) labeled by a field constant, then $order(f_v)=()$, where $()$ is the empty sequence. If $v$ is an input gate labeled by a variable $x_{i,j}$, then $order(f_v)=(i)$.
	
	\item If $v$ is an product gate, then $order(f_v)=order(f_u)\sqcup order(f_w)$, where the interval $order(f_v)$ is obtained by appending $order(f_u)$ with $order(f_w)$.
	
	\item If $v$ is a sum gate, then $order(f_v)=order(f_u)=order(f_w)$.  
\end{itemize}

One can define several versions of non-commutative $DET_n$ polynomial.
Non-commutative circuits computing the $DET_n$ polynomial, where the first index of the variables in each monomial is in increasing order, can be seen as regular set-multilinear w.r.t the identity permutation. 
In \cite{AS10}, it was shown that computing the non-commutative determinant polynomial is as hard as computing the commutative permanent polynomial. A natural next step is to find the set-multilinear circuit complexity of the commutative determinant polynomial.

We study the computational complexity of the commutative determinant polynomial $DET_n$ computed by a sum of regular set-multilinear circuits. We show that if the determinant polynomial $DET_n$ is computed by a circuit $C$ of size $s$, where $C$ is a sum of constantly-many regular set-multilinear circuits, then we can modify $C$  to compute the permanent polynomial $PERM_{n^\epsilon}$,where $\epsilon>0$, such that the new circuit size is polynomially related to the size of $C$. We remark that in our result, there is no restriction on the number of different parse tree types/shapes (see e.g., \cite{AR16}) allowed in each regular circuits.

One can view this as a generalization of the result shown in \cite{AS10} to a class of set-multilinear circuits computing the determinant polynomial $DET_n$. We obtain our result by carefully combining \emph{Erd{\"o}s-Szekeres theorem} \cite{ES35} and  some properties that we prove about regular set-multilinear circuits and the result of \cite{AS10}.

\section{Preliminaries}

\subsection{Determinant and Permanent}
\begin{definition}{(Commutative Determinant and Permanent)}
	Given the set of variables $X=\{x_{i,j} \mid 1 \leq i,j  \leq n\}$, the $n \times n$ commutative determinant and the $n \times n$ commutative permanent over $X$, denoted by $DET_n(X)$ and $PERM_n(X)$ respectively, are $n^2$-variate polynomials of degree $n$ given by:

	$$DET_n(X)=\sum_{\sigma \in S_n}sgn(\sigma)\prod_{i=1}^{n}x_{i,\sigma(i)} $$
	$$PERM_n(X)=\sum_{\sigma \in S_n}\prod_{i=1}^{n}x_{i,\sigma(i)}, $$
	
\end{definition}

Non-commutative determinant can be defined in various ways depending on the order in which variables are multiplied. One natural type of non-commutative determinant, called the \emph{Cayley determinant} $CDET_n$, is one where the order of multiplication is the identity permutation w.r.t first index of the variable.

\subsection{Erd{\"o}s-Szekeres Theorem}
\begin{theorem}[Erd{\"o}s-Szekeres Theorem, \cite{ES35}]
	Let $n$ be a positive integer. Let S be a sequence of distinct integers of length at least $n^2 + 1$. Then, there exists a monotonically increasing subsequence of S of length $n + 1$, or a monotonically decreasing subsequence of S of length $n + 1$.
\end{theorem}
Let $A,B$ be two $n \times n$ matrices. The following are known facts about the determinant and permutations.\\
{\bf Fact 1:} $det(A\times B)=det(A)\times det(B)$. \\
{\bf Fact 2:} The determinant of a permutation matrix is either +1 or -1.\\
{\bf Fact 3:} Let $\tau,\sigma \in S_n$. Then $sign(\tau \circ \sigma)=sign(\tau)\times sign(\sigma)$.\\

For $n \in \mathbb{N},$ let $[n]=\{1,2,\cdots,n\}$.

\section{Hardness of the Determinant: Sum of Two Regular Set-Multilinear Circuits}
In this section, we show that if the determinant polynomial is computed by a sum of two regular set-multilinear circuits then the permanent polynomial can also be represented as a regular set-multilinear circuit. This result involves all the techniques which will be used in the main result and it is easy to explain in this sum of two regular circuits model.
In the next section, we will prove the result for sum of constantly many regular set-multilinear circuits.
We note that all our polynomials are commutative. For the purpose of readability, we sometimes ignore the floor operation.

Let  $X=\{x_{i,j} \mid 1 \leq i,j  \leq n\}$ be the set of variables. Let $X_i=\{x_{ij}\mid 1\le
j\le n\}$ for $1\le i\le n$. Our aim is to show that if $C=C^{\sigma1}_1 + C^{\sigma2}_2$ computing the determinant polynomial $DET_n(X) \in \mathbb{F}[X]$, where the circuits $C^{\sigma1}_1$, $C^{\sigma2}_2$ are \emph{regular set-multilinear circuits} w.r.t $\sigma1,\sigma2 \in S_n$ respectively, then  there is an efficient transformation that converts the given circuit $C$ to another circuit $C'$ computing the permanent polynomial of degree $\sqrt{n}/2$.
Given $C=C^{\sigma1}_1 + C^{\sigma2}_2$ computing $DET_n(X)$, if $\sigma1=\sigma2$ then we can directly adapt the result of  \cite{AS10} and get a circuit $C'$ computing the permanent polynomial of degree $n/2$. If $n$ is not even then  we can substitute variables in the set $X_n$ suitably from $\{0,1\}$ such that $C$ computes $DET_{n-1}(X)$ before using the result of \cite{AS10}.

The case of $\sigma1 \neq \sigma2$ needs more work  that we explain now. The idea is to use the well known Erd{\"o}s-Szekeres Theorem \cite{ES35} that guarantees that any sequence of $n$ distinct integers contains a subsequence of length at least $\sqrt{n}$  that is either monotonically increasing or  decreasing. By viewing $\sigma=(\sigma(1),\sigma(2),\cdots,\sigma(n))$ as a sequence of integers, we apply the above result to permutations $\sigma1,\sigma2 \in S_n$. We first apply it to $\sigma1=(\sigma1(1),\sigma1(2),...,\sigma1(n))$ and let
$A=\{i_1,i_2,\cdots,i_{\sqrt{n}}\}$ be the set of indices that appear in this monotone subsequence. 
If the subsequence   is monotonically increasing then we do substitutions in  $DET_n(X)$ so that it computes the determinant polynomial of $\sqrt{n}\times\sqrt{n}$ matrix whose rows and columns are labeled by the elements of set $A$. This is done by making suitable substitutions to the variables in $X$ from $X \cup \{0,1\}$ in the given circuit $C$. After this we get a circuit $C'$ from $C$ that computes $DET_{\sqrt{n}}(X')$ where $X'=\bigsqcup_{i \in A}X_i$.

We note that $C'=C^{\sigma1'}_1 + C^{\sigma2'}_2$ where $\sigma1',\sigma2' \in S_{\sqrt{n}}$ and $\sigma1'=(\sigma1'(1),\sigma1'(2),\cdots,\sigma1'(\sqrt{n}))$ is in increasing order. 
If $\sigma1'=\sigma2'$, then we can use \cite{AS10} and get the permanent of degree $\sqrt{n}/2$. Otherwise,
we apply Erd{\"o}s-Szekeres Theorem to permutations $\sigma1',\sigma2'$. In particular, this will give us a monotone subsequence in $\sigma2'=(\sigma2'(1),\sigma2'(2),\cdots,\sigma2'(\sqrt{n}))$ with length at least $n^{1/4}$. If this sequence is  increasing, then the same subsequence is also  increasing in $\sigma1'$ as we already noted that it is in increasing order. Let $A_1=\{j_1,j_2,\cdots,j_{n^{1/4}}\}$ be the set of indices that appear in this monotone subsequence. Now we project, as before  so that it computes the determinant polynomial of a $n^{1/4}\times n^{1/4}$ matrix whose rows and columns are labeled by the elements in set $A_1$. After substituting from 
$X'\cup \{0,1\}$ for each variable in the given circuit  $C'$, we get a regular circuit $C^{''}$ that computes $DET_{n^{1/4}}(X'')$, where  $X''=\bigsqcup_{i \in A_1}X_i$. 

The important thing to note here is that in the new circuit $C^{''}=C^{\sigma1''}_1 + C^{\sigma2''}_2$, where $\sigma1'',\sigma2'' \in S_{n^{1/4}}$, both $\sigma1''$ and $\sigma2''$ are the same, i.e., $\sigma1''=\sigma2''$. 
We can rename the variable sets in $X''=\bigsqcup_{i \in A_1}X_i$ to $X_1,X_2,\cdots,X_{n^{1/4}}$. For example, if $i_1 \in A_1$ is the lowest index then we can rename $X_{i_1}$ to $X_1$, and for all $j$, rename $X_{i_1,j}$ to $X_{1,j}$. Similarly, the $k$-th lowest index is modified. After these modifications, we can assume that $\hat{X}=\bigsqcup_{i \in [n^{1/4}]}X_i$. 

As we noted before, any non-commutative circuit computing $DET_n$, where the first index of the variables in each monomial is in increasing order, can be seen as \emph{regular set-multilinear w.r.t identity permutation}.
Now we can apply the following theorem (Theorem 10 from \cite{AS10}) to get our result.
\begin{theorem}(Theorem 10, \cite{AS10})
	For any $n \in \mathbb{N}$, if there is a non-commutative circuit $C$ of size $s$ computing the Cayley determinant $DET_{2n}(X)$ then there is a circuit $C'$ of size polynomial in $s$ and $n$ that computes the Cayley permanent $PERM_{n}(Y)$.
\end{theorem}
 
If $n'=\lfloor n^{1/4} \rfloor $ is not an even number then we ignore the $X_{n'}$ variable set in $\hat{X}$ by following   substitutions: $X_{n',n'}=1$ and for all $j \in [n'-1]$, $X_{n',j}=0$ and $X_{j,n'}=0$. 
After this substitutions, we have a circuit that computes the determinant $DET_{n'-1}$ polynomial. 
Now applying the above theorem we get a circuit $\widehat{C}$ that computes the permanent polynomial of degree $\frac{n'-1}{2}$.

We now explain how to handle if Erd{\"o}s-Szekeres Theorem guarantees only monotonically decreasing sequence.  For that we define the \emph{reverse} of a regular set-multilinear circuit $C$ w.r.t $\sigma \in S_n$ computing a polynomial $f$. This results in a regular set-multilinear circuit $C^{rev}$ w.r.t $\sigma^{rev} \in S_n$, where $\sigma^{rev}=(\sigma(n),\sigma(n-1),...,\sigma(1))$, \emph{computing the same commutative polynomial $f$} as circuit $C$. We note that  if $\sigma$ has monotonically decreasing subsequence of length $k$ then $\sigma^{rev}$ has a monotonically increasing subsequence of \emph{same} length $k$. We obtain $C^{rev}$ by interchanging the left and right children of product gates in $C$. This is proved in the following lemma.

\begin{lemma}[Reversal Lemma]
\label{lem:reverse}
      Let $X=\{x_{i,j} \mid 1 \leq i,j  \leq n\}$ be a set of variables and $X=X_1 \sqcup X_2 \sqcup ... \sqcup X_n$ be a partition of $X$, where for all $1 \leq i \leq n$, $X_i=\{x_{i,1},x_{i,2},...,x_{i,n}\}$. Let $C$ be a regular set-multilinear  circuit w.r.t a permutation $\sigma \in S_n$ computing the polynomial $f\in F[X]$. 
      Then, there exists a regular set-multilinear circuit $C^{rev}$ w.r.t $\sigma^{rev} \in S_n$ where $\sigma^{rev}=(\sigma(n),\sigma(n-1),...,\sigma(1))$ computing the same commutative polynomial $f$ as circuit $C$. Moreover, the size of $C^{rev}$ is same as that of $C$.
\end{lemma}
    
    \begin{proof}
        First, we describe the construction of the circuit $C^{rev}$, and then prove its correctness. Let $v$ be a gate in $C$. As $C$ is a regular set-multilinear circuit w.r.t $\sigma \in S_n$, we have an interval $I_v$ w.r.t  the permutation $\sigma$ associated with the gate $v$. \\
        \textbf{Construction of $C^{rev}$:} Starting with the product gates at the bottom of $C$ and gradually moving up level-by-level, swap the left and right children of each product gate. \\
        \textbf{Correctness: } We show by induction on depth $d$ of $C$ that both circuits $C$ and $C^{rev}$ compute the same polynomial $f\in F[X]$ and $C^{rev}$ is a regular set-multilinear circuit w.r.t $\sigma^{rev} \in S_n$, where $\sigma^{rev}=(\sigma(n),\sigma(n-1),...,\sigma(1))$.
        Let $f_v$ and $f^{rev}_v$ denote the polynomials computed at any node $v$ in $C$ and $C^{rev}$, respectively. 
        Let $order(f_v)=I_v$. We will show that $f_v$ and $f^{rev}_v$ are the same polynomial and the only difference is in their orders. That is, $order(f^{rev}_v) = rev(order(f_v))$, where $rev(order(f_v))$ is $order(f_v)$ written in reverse (i.e., the interval $I_v$ is reversed).
        
        The proof is by induction  on the depth $d$ of the  circuit $C^{rev}$. Let $f^{rev}$ denote the polynomial computed by $C^{rev}$.
        
        \textbf{Base Case:} The base case is any node at depth 0, i.e., a leaf node. Consider any leaf node $l$. Then $f_l$, the polynomial computed at $l$, is either a variable or a field constant in $F$. If $f_l$ is a field constant, then $order(f_l) = ()$. Therefore, $order(f^{rev}_l)=()$. If $f_l$ is a variable $x_{i,j},1 \leq i,j \leq n$, then $order(f_l)=(i)$. Therefore, the $order(f^{rev}_l) = (i)$. In both cases, $f^{rev}_l=f_l$ and $order(f^{rev}_l)=rev(order(f_l))$.
        
        \textbf{Induction Hypothesis:} Assume for any node $u$ at depth $d'$, $1 \leq d' \leq d-1$, that $f^{rev}_u=f_u$ and $order(f^{rev}_u) = rev(order(f_u))$.
        
        \textbf{Induction Step:} Consider any node $v$ at depth $d'+1$, with $v_L$ and $v_R$  as its left and right children, respectively. By induction hypothesis, $f^{rev}_{v_L}=f_{v_L}$ and $order(f^{rev}_{v_L}) = rev(order(f_{v_L}))$. Similarly, $f^{rev}_{v_R}=f_{v_R}$ and $order(f^{rev}_{v_R}) = rev(order(f_{v_R}))$. 
        
        If $v$ is a product gate, then $f^{rev}_v=f^{rev}_{v_R} \times f^{rev}_{v_L}$, which is equivalent to  $f_{v_R} \times f_{v_L} = f_v$ by induction hypothesis. By induction hypothesis, $order(f^{rev}_v)$ is $order(f^{rev}_{v_R})$ appended with $order(f^{rev}_{v_L})$. The $order(f^{rev}_{v_L}) = rev(order(f_{v_L}))$, and $order(f^{rev}_{v_R}) = rev(order(f_{v_R}))$. Therefore,  $order(f^{rev}_v) = rev(order(f_v))$.

        If $v$ is a sum gate, then $f^{rev}_v=f^{rev}_{v_L} + f^{rev}_{v_R}$, which is equivalent to  $f_{v_L} + f_{v_R} = f_v$  by induction hypothesis. As $v$ is a sum gate, $order(f_v) = order(f_{v_L}) =order(f_{v_R})$.  As $order(f^{rev}_{v_L}) = rev(order(f_{v_L}))$ by induction hypothesis, we have that $order(f^{rev}_v) = rev(order(f_v))$ and $order(f^{rev}_{v_R}) = rev(order(f_{v_R}))$. Thus, 
        $order(f^{rev}_v) = order(f^{rev}_{v_L}) = order(f^{rev}_{v_R})$.  
        
        The size of $C^{rev}$ is same as that of $C$ because the only modification we are doing to $C$ is swapping the children of product gates. This completes proof of the lemma.      
    \end{proof}
Using Lemma \ref{lem:reverse}, we can handle the monotonically decreasing sequence without modifying the polynomial computed by a regular set-multilinear circuit. This gives us a circuit $\widehat{C}$ that computes the permanent polynomial of degree $\frac{\sqrt[4]{n}}{2}$. We remark that Lemma \ref{lem:reverse} can be adapted for non-commutative circuits as well. 

We now explain how to get the permanent polynomial of degree $\frac{\sqrt{n}}{2}$ instead of $\frac{\sqrt[4]{n}}{2}$. This gives us quadratic improvement in the degree of the permanent polynomial.  This is based on the observation that if $C$ is a regular set-multilinear  circuit w.r.t a permutation $\sigma \in S_n$ computing the determinant polynomial $DET_n(X)$, then for any permutation $\tau \in S_n$, there is another regular set-multilinear circuit $C'$  w.r.t $\tau \circ \sigma \in S_n$ computing the same determinant polynomial $DET_n(X)$. Moreover, the size of $C'$ is at most one more than the size of $C$. 

In other words, composition of permutations can be efficiently carried out for regular set-multilinear  circuits computing the determinant polynomial $DET_n(X)$. 
\begin{lemma}[Composition Lemma]
\label{lem:composition}
Let $C=C_1 + C_2$ be the sum of two regular set-multilinear circuits computing the determinant polynomial $DET_n(X)$, where the circuits $C_1$, $C_2$ are \emph{regular set-multilinear circuits} w.r.t $\sigma_1,\sigma_2 \in S_n$ respectively. Then for any permutation $\tau \in S_n$, there exists another circuit $C'$ that computes $DET_n(X)$. $C'$ is also a sum of two regular set-multilinear circuits (regular set-multilinear w.r.t $\tau \circ \sigma1, \tau \circ \sigma2  \in S_n$).  Moreover, the size of $C'$ is at most one more than the size of $C$.
    \end{lemma}
    
    \begin{proof}
     First, we describe the construction of the circuit $C'=C'_1+C'_2$ and then prove its correctness. \\
        \textbf{Construction of $C^{'}$:}  For every variable $x_{i,j}$ in $C=C_1 + C_2$, substitute the variable $x_{\tau(i),j}$. 
        Let $\widehat{C}$ be this modified circuit. 
        If $sgn(\tau)$ is -1, then add a leaf node labeled -1 and multiply the root node of $\widehat{C}$ with this leaf node.  Let $C'$ be this modified circuit. The size of $C'$ is at most one more than the size of $C$.\\
         \textbf{Correctness: }Now we will prove that $C'$ computes $DET_n(X)$. Let $m_1$ and $m_2$ be any two monomials in $DET_n(X)$ computed by $C$. Let $m'_1$ and $m'_2$ be the monomials obtained by applying $\tau$ to the first index of each of the variables in $m_1$ and $m_2$ respectively. The permutations corresponding to $m'_1$ and $m'_2$ are $\tau\circ\sigma_1$ and $\tau\circ\sigma_2$ respectively. 
         \begin{itemize}
          \item  Case 1: $m_1 = m_2$.  We show that $m'_1=m'_2$ in  $C'$. We note that $m_1$ and $m_2$ could be computed by circuits $C_1$ and $C_2$ respectively. Thus, the order of variables appearing in $m_1$ and $m_2$ could be different in general. By construction of $C'$, $x_{i,j}$ is substituted by the variable $x_{\tau(i),j}$. Since $m_1 = m_2$, we get $m'_1=m'_2$.
          \item Case 2: $m_1 \neq m_2$. We show that $m'_1 \neq m'_2$ in $C'$.  Since $m_1 \neq m_2$, there exists a variable $x_{i_1,j_1}$ in $m_1$ and a variable $x_{i_2,j_2}$ in $m_2$ such that $x_{i_1,j_1} \neq x_{i_2,j_2}$. Suppose $j_1 = j_2$, then $i_1 \neq i_2$. Then, $x_{\tau(i_1),j_1} \neq x_{\tau(i_2),j_2}$. This implies $m'_1 \neq m'_2$. Suppose $j_1 \neq j_2$, then $x_{\tau(i_1),j_1} \neq x_{\tau(i_2),j_2}$, which again implies that $m'_1 \neq m'_2$. 
         \end{itemize}

         By construction of $C'$, we note that coefficients of monomials are not affected.  Now we will prove that $C'$ computes $DET_n(X)$. Let $A_X$ be a $n \times n$ matrix where row $i$ contains all variables of the set $X_{i}$. In other words, the entry of $i$-th row and $j$-th column of the matrix $A_X$ is $x_{i,j}$. Let $\beta \in S_n$. By changing $x_{i,j}$ to $x_{\beta(i),j}$, in effect it permutes the rows of $A_X$. In other words, the determinant is equal to the determinant of $P_{\beta}\times A_X$, where $P_{\beta}$ is the $n \times n$ permutation matrix. The entry of $i$-th and $j$-th column of $P_{\beta}$ is 1 iff $j=\beta(i)$ and 0 otherwise. By Fact 1 and 2, we have $det(P_{\beta}\times A_X)=det(P_{\beta})\times det(A_X)=sign(\beta) \times det(A_X).$

         Thus, composing the permutation $\tau$ with $\sigma_1,\sigma_2$ maps different monomials to different monomials and in effect does not change the determinant computed except that the sign changes.
      Note that $sgn(\tau\circ \beta)=sgn(\tau).sgn(\beta)$ (by Fact 3). Therefore, if $sgn(\tau)=-1$, then the coefficients of $m'_1$ and $m'_2$ are the negatives of the coefficients of $m_1$ and $m_2$ respectively.  Therefore, if $sgn(\tau)=-1$, $C'$ computes $DET_n(X)$, as the leaf gate labeled -1 multiplied to the output gate ensures that $C'$ computes $DET_n(X)$. However, the coefficients of $m'_1$ and $m'_2$ are the same as the coefficients of $m_1$ and $m_2$ respectively, if $sgn(\tau)=+1$. In the case that $sgn(\tau)=+1$, there is no need of this leaf gate. In both cases, the polynomial computed by $C'$ is $DET_n(X)$. 
        
        Now we will show that $order(C_j)=(\tau(\sigma_j(1)),\tau(\sigma_j(2)),...,\tau(\sigma_j(n)))$, $j \in \{1,2\}$. The proof is by induction on the depth $d$ of the circuit. We will prove it for $C_1$. The proof is similar for the circuit $C_2$. Recall that $C_1$ is regular set-multilinear circuit w.r.t $\sigma_1$. Let $v$ be a gate in the circuit. We denote polynomial computed at $v$ in $C$ and $C'$  by $f_v$ and $f'_v$ respectively.
        
        \textbf{Base Case:} The base case is any node at depth 0, i.e, a leaf node. Let $\ell$ be any leaf node. Then $f_\ell$ is either a field constant or a variable $x_{i,j}$. If $f_\ell \in F$, then the $order(f_\ell)$ is the empty sequence $()$. As there is no variable in $f_\ell$, there is no change to be made. Therefore, $order(f'_\ell) = ()$, and therefore the claim trivially holds. If $f_\ell$ is a variable $x_{i,j}$, then $order(f_\ell)=(i)=(\sigma_1(k)),$ for some $k \in \{1,2,...,n\}$. We change $x_{i,j}$ to $x_{\tau(i),j}$, which means $order(f'_\ell)=(\tau(\sigma_1(k)))$. 
        
        \textbf{Induction Hypothesis:} Suppose the claim holds for any node at depth $d', 1 \leq d' < d$.
        
        \textbf{Induction Step:} Consider any node $v$ at depth $d'+1$. Let $u$ and $w$ be its left and right children with degrees $d_u,d_w$ respectively. 
        \begin{itemize}
         \item Case 1: $v$ is a sum gate. Thus, $f'_v=f'_u + f'_w$. Then $order(f'_u)=order(f'_w)=order(f'_v)$. 
         \item Case 2: $v$ is a product gate. Thus, $f'_v=f'_u \times f'_w$. Let $0 \le a \le n-d_u-d_w$, where $d_u,d_w$ denote degrees of $f_u,f_w$ respectively.\\      Let $order(f_u)=(\sigma_1(a+1),\sigma_1(a+2),\cdots,\sigma_1(a+d_u))$ and \\$order(f_w)=(\sigma_1(a+d_u+1),\sigma_1(a+d_u+2),\cdots,\sigma_1(a+d_u+d_v))$.
        By IH, $order(f'_u)=(\tau(\sigma_1(a+1)),\tau(\sigma_1(a+2)),...,\tau(\sigma_1(a+d_u)))$, and let $order(f'_w)=(\tau(\sigma_1(a+d_u+1)),\tau(\sigma_1(a+d_u+2)),...,\tau(\sigma_1(a+d_u+d_v)))$. Then $order(f'_v)=(\tau(\sigma_1(a+1)),\cdots,\tau(\sigma_1(a+d_u)),\tau(\sigma_1(a+d_u+1)),\cdots,\tau(\sigma_1(a+d_u+d_v)))$.
        \end{itemize}
  Thus, in both cases, the claim holds. This completes the proof of the lemma.
    \end{proof}
    Unlike Lemma \ref{lem:reverse}, we note that in general this composition operation may not hold for any polynomial $f$ computed by a regular circuit. For example, if $C$ is a regular set-multilinear circuit computing the polynomial $f=x_{1,1}x_{2,0}x_{3,0}x_{4,1}$ then by swapping the 3rd and 4th indices, we get a different polynomial  $f'=x_{1,1}x_{2,0}x_{4,0}x_{3,1}$. Now we have all results needed to the case where the determinant polynomial is computed by a sum of two regular set-multilinear circuits.
   
    \begin{theorem}
  Let $X=\{x_{i,j}\}_{i=1,j=1}^n$. If the determinant polynomial over $X$ is computed by a circuit $C$ of size $s$, where $C$ is the sum of two regular set-multilinear circuits, then the permanent polynomial of degree $\sqrt{n}/2$ can be computed by a regular set-multilinear circuit $C'$ of size polynomial in $n$ and $s$.
\end{theorem}

\begin{proof}
Let $C=C^{\sigma_1}_1 + C^{\sigma_2}_2$, where the circuits $C^{\sigma_1}_1$, $C^{\sigma_2}_2$ are \emph{regular set-multilinear circuits} w.r.t $\sigma_1,\sigma_2 \in S_n$ respectively. We show that there is an efficient transformation that converts the given circuit $C$ to another circuit $C'$ computing the permanent polynomial of degree $\sqrt{n}/2$. 

Without loss of generality, we can assume that $\sigma_1$ is the identity permutation. This is because otherwise by Lemma \ref{lem:composition}  we can get a new circuit $\hat{C}=C^{\sigma_1^{-1}\circ \sigma_1}_1 + C^{\sigma_1^{-1}\circ\sigma_2}_2$ with $\sigma_1^{-1}\circ \sigma_1, \sigma_1^{-1}\circ\sigma_2 \in S_n$ as the two permutations used. This does not increase the circuit size. By the Erd{\"o}s-Szekeres Theorem, there is a  monotone subsequence of length $\sqrt{n}$.
 Let  $A$ be the set of all such indices. 
 \begin{itemize}
  \item Case 1: Subsequence is increasing. As $\sigma_1$ is the identity, the same subsequence of indices in $\sigma_1$ is also increasing.  We do the following substitutions. For all $j \notin A$, set $x_{j,j}=1$ and  for all $i \in [n]$ and $i \neq j$, set $x_{j,i}=0$ and $x_{i,j}=0$.  After this substitutions, the circuit computes the determinant polynomial over  $A'=\bigsqcup_{i \in A }X_i$ and the order of the subsequence in both $C_1$  and $C_2$ are the same.
  We rename the variable sets in $A'$ as follows: if $i_1 \in A_1$ is the $j$-th lowest index in the subsequence then we  rename $X_{i_1}$ to $X_j$, and for all $k$, rename $X_{i_1,k}$ to $X_{j,k}$. The modified circuit $C'$ computes the determinant polynomial over $\hat{X}=\bigsqcup_{i \in [n^{1/2}]}X_i$ and it is regular w.r.t the identity permutation in $S_{\sqrt{n}}$.

 \item Case 2: Subsequence is decreasing. Then by Lemma \ref{lem:reverse}, we modify  the circuit $C^{\sigma2}_2$  to get a new circuit computing the same polynomial as computed by the circuit $C^{\sigma_2}_2$ but the new circuit is regular set-multilinear w.r.t the permutation $\sigma_2^{rev}=(\sigma_2(n),\sigma_2(n-1),\cdots, \sigma_2(1))$. We note that, by applying Lemma \ref{lem:reverse}, no sign change occurs to the determinant polynomial. 
 In this modified (second) circuit, the corresponding  subsequence now becomes increasing.  This reduces this case to case 1.  
  \end{itemize}
  
 Thus, after this modifications we have a new regular circuit $C'$, that computes the determinant polynomial of degree $\sqrt{n}$, w.r.t the identity permutation. 
 If $\lfloor \sqrt{n} \rfloor$ is not an even number then we substitute variables in $X_{\sqrt{n}}$ as explained before. Thus, $C'$ computes the determinant polynomial of even degree. Now by the result of \cite{AS10}, we can compute the permanent polynomial of degree $\frac{\sqrt{n}}{2}$ by a circuit of size polynomial in $s$ and $n$.  This completes the proof of the theorem.
\end{proof}

\section{Hardness of the Determinant: Sum of Constantly-Many Regular Set-Multilinear circuits}
 In this section, we show that if the determinant polynomial $DET_n(X)$ is computed by a sum of constantly many regular set-multilinear circuits then the permanent polynomial $PERM_{n^\epsilon/2}(X)$, $\epsilon>0$ depends on $k$, computed a regular circuit. The proof of the following lemma is omitted due to lack of space. This is a generalization of the  (composition) Lemma \ref{lem:composition} but idea of the proof is similar.
 
  \begin{lemma}
     Let $C=C_1+C_2+\cdots+C_k$ be a sum of $k$ regular set-multilinear circuits such that $C$ computes $DET_n(X)$. Let $C_1,C_2,...,C_k$ be regular set-multilinear w.r.t $\sigma_1,\sigma_2,...,\sigma_k$ respectively, where each $\sigma_i \in S_n$. For any $\tau \in S_n$, let $C_1^{\tau(\sigma_1)}, C_2^{\tau(\sigma_2)}, ...., C_k^{\tau(\sigma_k)}$ be the circuits obtained by substituting $x_{\tau(i),j}$ for each variable in $x_{i,j}$ in each of the $k$ circuits. Let $C'$ be the sum of $C_1^{\tau(\sigma_1)}, C_2^{\tau(\sigma_2)}, ...., C_k^{\tau(\sigma_k)}$. Then $C'$ also computes $DET_n(X)$. Moreover, the size of $C'$ is at most one more than the size of $C$.    
    \end{lemma}

 Without loss of generality we can assume that for each $i\neq j \in [k]$, $\sigma_i \neq \sigma_j$. Otherwise, we can combine all $C_i$'s which use same $\sigma_i$ into a single $C_i$ using addition gates and get a circuit $C$ that is a sum of $k'$ regular set-multilinear circuits, where $k'<k$. Therefore, $C$ is the sum of $k'$ regular set-multilinear circuits such that no two permutations used by any two of these $k'$ circuits is same. We call such a circuit $C$ as \emph{$k'$-regular circuit}.
 
  \begin{theorem}
 Let $C$ be the sum of $k$-many regular set-multilinear circuits, of size $s$, computing the determinant polynomial $DET_n(X)$. Then there exists a regular set-multilinear circuit whose size is at most $s+1$ that computes the determinant polynomial $DET_{n^{\epsilon}}(X')$, where $X'=\{x_{i,j}\}_{i=1,j=1}^{n^{\epsilon}}$ and  $\epsilon \geq 1/2^{k-1}$.
\end{theorem}
\begin{proof} 
	Let $C=C^{\sigma_1}_1 + C^{\sigma_2}_2+\cdots+C^{\sigma_k}_k$, where the circuits $C^{\sigma_i}_i$ are regular set-multilinear circuits w.r.t $\sigma_i \in S_n$, $i \in [k]$. We show that there is an efficient transformation that converts the given circuit $C$ to another circuit $C'$ computing the determinant polynomial of degree $n^\epsilon$, $\epsilon = 1/2^{k-1}$ . 
	
	Without loss of generality, we can assume that $\sigma_1$ is the identity permutation. This is because otherwise by Lemma \ref{lem:composition}  we can get a new circuit $\hat{C}=\hat{C}_1 + \hat{C}_2 + \cdots + \hat{C}_{k'}$ 
	where $\hat{C}_i$ is a regular set-multilinear circuit w.r.t the permutation $\sigma_1^{-1}\circ \sigma_i \in S_n$, where $i \in [k]$. We note that $\hat{C}$ computes the same polynomial as circuit $C$ and both circuits have the same size.

	Denote by $C^{(\ell)}$ the circuit obtained after the $\ell$-th iteration, where $0 \leq \ell < k$. We will show that $C^{(\ell)}$ computes the determinant polynomial of degree  $n^{1/2^\ell}$  and $C^{(\ell)}$ is a $(k-\ell)$-regular circuit.

	At iteration 0, this condition holds, as $C^{(0)}=C $ computes the determinant polynomial over $X$ and $C^{(0)}$ is a $k$-regular circuit.

	Suppose the condition is true for some $m$, where $0 \leq m < k$. We will show that $C^{(m+1)}$ computes the determinant polynomial of degree $n^{1/2^{m+1}}$ and $C^{(m+1)}$ is a $k-(m+1)$-regular circuit. Note that $C_1,C_2,\cdots,C_{k}$ have been modified during the first $m$ iterations. Let us denote these modified circuits at the end of the $m$-th iteration by $C'_1,C'_2,\cdots,C'_{k}$. Thus, $C^{(m)}=C'_1+C'_2+\cdots+C'_{k}$.
	
	Without loss of generality, we will assume that each variable in the determinant computed by $C^{(m)}$ has both its indices in $X^{(m)}=\{1,2,\cdots,k_m\}$, where $k_m = n^{\frac{1}{2^m}}$. We note that the first $m$ regular set-multilinear circuits $C'_1,C'_2,\cdots,C'_m$ are regular w.r.t identity permutation $id \in S_{k_m}$. 
	As noted before, we can combine all $C'_i$'s which has same $\sigma_i$ as single $C_i$ using addition gates.
	By Erd{\"o}s-Szekeres Theorem \cite{ES35}, in $\sigma'_{m+1}$, there is a monotone subsequence of length  $n^{\frac{1}{2^{m+1}}}$. There are two cases to handle based on whether the subsequence is increasing or decreasing.
	
	\begin{itemize}
		\item Case 1: Suppose $\sigma'_{m+1}$ has an increasing subsequence. Let $S^{(m+1)}=\{i_1,i_2,\cdots,i_{k_{m+1}}\}$ be the set of indices in this increasing subsequence, where $k_{m+1} = n^{\frac{1}{2^{m+1}}}$. We do the following substitutions. For all $j \notin S^{(m+1)}$, set $x_{j,j}=1$ and  for all $i \in [k_m]$ and $i \neq j$, set $x_{j,i}=0$ and $x_{i,j}=0$.  After these substitutions, the circuit computes the determinant polynomial over  $A'=\bigsqcup_{i \in S^{(m+1)} }X_i$.
  We rename the variable sets in $A'$ as follows: if $i_1 \in S^{(m+1)}$ is the $j$-th lowest index in the subsequence then we  rename $X_{i_1}$ to $X_j$, and for all $k$, rename $x_{i_1,k}$ to $x_{j,k}$. The modified circuit $C^{(m+1)}$ computes the determinant polynomial over $\hat{X}=\bigsqcup_{i \in [k_{m+1}]}X_i$. It is clear that $\sigma'_1=\sigma'_2=\cdots=\sigma'_m=\sigma'_{m+1}=identity$. This shows that  $C^{(m+1)}$ is a $k-(m+1)$-regular circuit.
	      \item Case 2: Suppose $\sigma'_{m+1}$ has only a decreasing subsequence,
	      then, we modify  the sub-circuit $C'_{m+1}$ by Lemma \ref{lem:reverse}  to get a new circuit computing the same polynomial as computed by the $(m+1)$-th sub-circuit in the previous iteration but the new circuit is regular set-multilinear w.r.t the permutation $\sigma_{m+1}^{rev}=(\sigma'_{m+1}(k_m),\sigma'_{m+1}(k_m-1),\cdots, \sigma'_{m+1}(1))$. Note that after reversal operation, Lemma \ref{lem:reverse} guarantees that the polynomial computed by the circuit $C'_{m+1}$ does not change. In $\sigma_{m+1}^{rev}$, the corresponding  subsequence now becomes increasing. It is clear that the same sequence of indices in $\sigma'_1,\sigma'_2,\cdots,\sigma'_m$ are also increasing. This reduces this case to case 1.
		 
	\end{itemize}

Clearly, $C^{(m+1)}$, obtained at the end of the $(m+1)$-th iteration, computes the determinant over $X^{(m+1)}=\{x_{i,j} \mid  i,j \in S^{(m+1)}\}$.  This implies that at the end of $(k-1)$-th iteration, $C^{(k-1)}$ computes the determinant of  degree $n^{\epsilon}$ over $X^{(k-1)}$, where $\epsilon=1/2^{k-1}$. Moreover, $C^{(k-1)}$ is a $1$-regular set-multilinear circuit. This completes the proof of the theorem. 
\end{proof}

 Let $d$ be the degree of the determinant polynomial computed by the circuit $C^{(k-1)}$ in the above theorem. Clearly, $d \geq n^{\frac{1}{2^{k-1}}}$. If $\lfloor d \rfloor $ is not an even number then like before we substitute variables in the set $X_{\lfloor d \rfloor}$ such that the modified circuit computes the determinant of even degree  $\lfloor d \rfloor-1$.  Now by the result of \cite{AS10}, we can compute the permanent polynomial of degree $d/2$ by a circuit of size polynomial in $s$ and $n$. 
 Thus, we get the following main result as a corollary.
 \begin{corollary}
    Let $C$ be the sum of $k$-many regular set-multilinear circuits computing the determinant polynomial $DET_n(X)$. Let $s$ denote the size of the circuit $C$. Then there exists a regular set-multilinear circuit $\widehat{C}$ computing the permanent polynomial $PERM_{n^{\epsilon}/2}$, where $\epsilon=1/2^{k-1}$.
  Moreover, the size of $\hat{C}$ is polynomial in $s$ and $n$.
 \end{corollary}

We note that to compute the permanent polynomial of degree $n$, we need to consider the determinant polynomial of degree $n^{2^{k-1}}$ computed by a $k$-regular circuit. So, our methods need $k$ to be a constant.

\section{Discussion}
In this paper we studied the complexity of computing the determinant polynomial using sum of constant number of regular set-multilinear circuits. We showed that computing the determinant in this model is at least as hard as computing the commutative permanent polynomial.
An interesting open question is whether our results can be extended to the sum of a non-constant (some function of the degree of the determinant) number of regular set-multilinear circuits. Another question is: What is  the complexity of computing the determinant polynomial using set-multilinear circuits?. This question was also raised in \cite{AR16}.


\end{document}